\documentclass[12pt,a4paper]{article}
\usepackage{amssymb,amsmath,amsthm}
\usepackage{epsfig}
\usepackage{tikz}
\usepackage{setspace}
\newtheorem{definition}{Definition}[section]
\newtheorem{proposition}[definition]{Proposition}

\newtheorem{theorem}[definition]{Theorem}
\usepackage{pstricks}
\usepackage{pstricks-add}
\usetikzlibrary{shapes}

\def\checkmark{\tikz\fill[scale=0.4](0,.35) -- (.25,0) -- (1,.7) -- (.25,.15) -- cycle;}

\newtheorem{example}[definition]{Example}
\newtheorem{remark}[definition]{Remark}

\newcommand{\RR}{\mathbb{R}}
\newcommand{\BB}{\mathbb{B}}
\DeclareMathOperator*{\argmax}{arg\,max}

\newcommand{\Img}{\operatorname{Im}}
\newcommand{\fix}{\operatorname{fix}}
\newcommand{\nonfix}{\operatorname{non-fix}}
\newcommand{\maj}{\operatorname{maj}}
\newcommand{\imax}[1]{{#1}\textup{-}\!\max}
\newcommand{\iargmax}[1]{{#1}\textup{-}\!\argmax}

\newcommand{\power}[1]{{\mathcal P}(#1)}
\newcommand{\ignore}[1]{}

\title{Higher-Order Game Theory\footnote{\scriptsize We thank the seminar participants at the University of Mannheim, the Dagstuhl Perspectives Workshop ``Categorical Methods at the Crossroads'', the Dagstuhl Seminar ``Coalgebraic Semantics of Reflexive Economics'', the ``Computing in Economics and Finance Conference 2014'' in Oslo, the ``Cogrow'' Workshop in Nijmegen 2014 and the ``Logics for Social Behavior'' in Den Haag 2014 for helpful comments. Hedges thanks EPSRC, grant EP/K50290X/1, for financial support. Oliva gratefully acknowledges financial support by the Royal Society through grant 516002.K501/RH/kk. Winschel and Zahn gratefully acknowledge financial support by the Deutsche Forschungsgemeinschaft (DFG) through SFB 884 ``Political Economy of Reforms''.}}
 
\author{
Jules Hedges, Paulo Oliva\\ 
\footnotesize School of Electronic Engineering and Computer Science, Queen Mary University London
\vspace{0.01cm}\\
Evguenia Sprits, Philipp Zahn\\ 
\footnotesize Department of Economics, University of Mannheim
\vspace{0.01cm}\\
Viktor Winschel\\ 
\footnotesize Department of Management, Technology and Economics, ETH Z\"urich
}
\date{\today}
\begin{document}
\maketitle
\begin{abstract}
\footnotesize
In applied game theory the motivation of players is a key element. It is encoded in the payoffs of the game form and often based on utility functions. But there are cases were formal descriptions in the form of a utility function do not exist. In this paper we introduce a representation of games where players' goals are modeled  based on so-called \emph{higher-order functions}. Our representation provides a general and powerful way to mathematically summarize players'
intentions.  In our framework utility functions as well as preference relations are special
cases to describe players' goals.  We show that in higher-order functions formal descriptions of players may still exist where utility functions do not using a classical example, a variant of Keynes' beauty contest. We also show that equilibrium conditions based on Nash can be easily adapted to our framework. Lastly, this framework serves as a stepping stone to powerful tools from computer science that can be usefully applied to economic game theory in the future such as computational and computability
aspects.
\end{abstract}
\textbf{JEL codes:} C0, D01, D03, D63, D64\\
\textbf{Keywords:}  behavioral economics, foundations of game theory, process orientation, computable economics, beauty contest, coordination and differentiation goals, higher order functions, quantifiers, selection functions

%\onehalfspacing
%---------------------------------------
\section{Introduction}
%---------------------------------------

A key ingredient of game theory is the \emph{motivation} or \emph{goal} of each individual player. 
The standard approach is to \emph{encode} the goal of the player in the payoffs assigned to the different outcomes of the interaction. 
These payoffs, in turn, are often based on utility functions, as for instance in political economics where a politician maximizes his vote share, or in industrial organization 
where a company maximizes profits. Yet, not all motivations can be
summarized easily or in the most insightful way by utility functions.

Consider the following example of a variant of Keynes' beauty contest \cite{Keynes1936}. There are three players, the judges $J=\{J_1, J_2, J_3\}$. Each judge simultaneously votes for one of two contestants, $A$ or $B$. The winner is determined by the simple majority rule of type $\operatorname{maj}: X\times X\times X \rightarrow X$, where $X = \{ A , B \}$. Note that $X$ in this case is both the set of possible choices for the three judges and the space of possible outcomes.

Suppose that $J_1$ prefers $A$ over $B$. Judges $2$ and $3$, in contrast, are not interested in the outcomes per se but only in voting for the winner of the contest. As long as the winning contestant is the one they voted for they are happy. How would we describe such a game? In particular, how would we assign payoffs to the outcomes of the contest? There are only two possible results: either $A$ wins or $B$ wins. Obviously, these outcomes do not contain sufficient information to model the payoffs of players $2$ and $3$, as both judges do not care about $A$ or $B$ but only whether there is a majority for the candidate they have voted for. 

One solution is to consider an extended outcome space, i.e. the complete voting profile of all players. Then, judge $2$ and $3$ prefer profiles where they are in the majority over profiles where they are in the minority. Table \ref{tab:voting} depicts the payoff matrix that represents this game. 

\begin{table}[!htb]\label{tab:voting}
     \begin{minipage}{.5\linewidth}
                  \centering
                  \begin{tabular}{l|ll}
                    & A & B   \\
                    \hline
                  A & 1, 1, 1 & 1, 0, 1 \\
                  B & 1, 1, 1 & 0, 1, 0 
                  \end{tabular}
                  \footnotesize
                  \emph{$A$} 
      \end{minipage}%
      \begin{minipage}{.5\linewidth}
                 \centering
                 \begin{tabular}{l|ll}
                    & A & B  \\
                    \hline
                  A & 1, 1, 0 & 0, 1, 1 \\ 
                  B & 0, 0, 1 & 0, 1, 1 \\
                 \end{tabular}
                 \footnotesize
                 \emph{$B$}
      \end{minipage} 
      \caption{Voting Contest}
\end{table}

The crucial point is that all the individual outcomes have to be ranked by hand and payoffs are assigned accordingly. The goals of the three players are \emph{implicitly encoded} in the payoffs. A succinct
formal summary of goals like a utility function is not provided.  

Now, ranking the extended outcome space is simple and fast in this example. But what if we considered a game consisting of 5 players or 20 players with players having various goals? Calculating the payoffs will become complicated, slow or error prone. More importantly, \emph{interpreting} (or \emph{decoding}) the outcome of such a game in terms of the original motivation becomes complicated. 

In this paper we introduce a representation of games that formally summarizes the goals of agents, such as in the beauty contest described above, in a general and powerful way. Our approach builds  on higher-order functions (also called functionals or operators) and originates in game-theoretic approaches to proof theory \cite{escardo_sequential_2011,escardo10a}.

The core concept is that we model players' goals as \emph{quantifiers}, i.e. higher-order functions of type $(X \to R)\to R$, that take functions from the space $X\to R$ as an argument, where $X$ is the set of choices and $R$ is the set of possible outcomes. A corresponding notion is that of a \emph{selection function}, i.e. a higher-order function of type $(X \to R) \to X$. If we think of functions of type $X \to R$ as a \emph{game context}, then quantifiers map game contexts to preferred outcomes, while selection functions map game contexts to preferred moves.

The operators normally used in game theory are examples of quantifiers and selection functions of a particular type. They are the  
$\max \colon (X \to \RR) \to \RR$ and $\argmax \colon (X \to \RR) \to X$ operators.

One point of view is that our approach generalizes these operators, and allows the use of different operators to model players, i.e. we can describe behavior in terms of operators other than $\max$.
Moreover, the outcome space in our formulation of goals can have any structure and is not restricted in order to be representable by rational preferences\footnote{In a companion paper \cite{Hedges_et_al_2015_decisions} where we focus on decision problems, we give a detailed description in which way quantifiers and selection functions are a common structure shared by utility  maximization, preference relations as well as behavioral approaches.}.

But most importantly, since quantifiers and selection functions take functions as input, context-dependent goals about the interaction itself, as in the Keynes example, can be easily described. For instance, in the case of the voting contest, judges $2$ and $3$ can be directly summarized as a fixed point operator on the outcome function. It is this flexibility  of our approach that allows for a high-level description of individual goals in games. 

Economic situations often encompass high-level goals. Consider, for instance, coordination or differentiation. Typically, we would translate these goals into numerical payoffs. While in some situations payoffs may be a natural representation of the original goals, in other cases this translation blurs the properties of the outcomes and alternatives to decide upon and thereby obfuscates the actual motivations of the players. By using quantifiers and selection functions we gain a more expressive description that is closer to the economic situation we want to encode as well as easier to adapt to alternative specifications of the game. 

We will also show that the standard Nash equilibrium concept can be seamlessly generalized to our higher-order representation of games. Based on Nash, we introduce the \emph{quantifier} and \emph{selection equilibrium} concepts. We prove that quantifier and selection equilibria coincide in the case of the classical max and argmax operators, but that, generally, this equivalence does not hold: For other quantifiers and selection functions the two different equilibrium concepts yield different sets of equilibria. We give a sufficient condition for the two notions to coincide based on the notion of \emph{closedness} of selection functions. We prove that in general, the selection equilibrium is an equilibrium refinement of the quantifier equilibrium, and present evidence that for games based on non-closed selection functions, the selection equilibrium is the appropriate solution concept (Section \ref{sec:fixpoint}).

We consider several variants of the voting contest introduced above in order to acquaint the reader with our tools. We also provide several textbook examples such as the Battle of the Sexes, and show that the underlying structure of such examples is accurately represented in this higher-order setting.

In this paper we restrict ourselves to simple models in order to highlight the features of our framework. 
As there are an infinite number of higher-order functions like the max operator,
we can only peek into the additional possibilities of higher-order functions and leave the applications for future research.

%---------------------------------------
\section{Players, Quantifiers and Selection Functions}
\label{sec:quantifiers}
%---------------------------------------

A \emph{higher order function} (or \emph{functional}) is a function whose domain is itself a set of functions. 
Given sets $X$ and $Y$ we denote by $X \to Y$ the set of all functions with domain $X$ and codomain $Y$. 
A higher order function is therefore a function $f : (X \to Y) \to Z$ where $X$, $Y$ and $Z$ are sets.

\begin{example} There are familiar examples of higher-order functions in economics. The $\max$ operator has type
\[ \max \colon (X \to \RR) \to \RR \]
returning the maximum value of a given real-valued function $p \colon X \to \RR$. One will normally write $\max p$ as $\max_{x \in X} p(x)$. A corresponding operator is $\argmax$ which returns all the points where the maximum is attained, i.e.
\[ \argmax \colon (X \to \RR) \to \power{X} \]
using $\power{X}$ for the power-set of $X$. As opposed to $\max$, $\argmax$ is naturally a multi-valued function. Integration has a similar type to $\max$
\[ \int \colon (X \to \RR) \to \RR \]
except that one will only consider measurable sets $X$, and $\int p$ is normally written as $\int p(x) dx$. Of a slightly different nature is the fixed point operator
\[ \fix \colon (X \to X) \to \power{X} \]
which calculates all the fixed points of a given self-mapping $p \colon X \to X$, or the anti-fixed-point operator which calculates all points that are not fixed points. 
\end{example}

In this section we define two particular classes of higher-order functions: \emph{quantifiers} and \emph{selection functions}. We first establish that these functions provide means to represent agents' goals in an abstract and general way. In particular,  these notions usefully generalize utility maximization and preference relations. 

%---------------------------------------
\subsection {Game Context}
\label{sec:context}
%---------------------------------------
To define players' goals we first need a structure that represents the economic situation on which these goals are based. For this end we introduce the concept of a \emph{game context} which summarizes information of the strategic situation from the perspective of a single player.

\begin{definition}[Game context] For a player $\mathcal A$ choosing a move from a set $X$, having in sight a final outcome in a set $R$, we call any function $p : X \to R$ a possible \emph{game context} for the player $\mathcal A$. 
\end{definition}

Assuming that the player $\mathcal A$ is \emph{deterministic} (or \emph{predictable}) in the sense that his moves are not dependent on chance\footnote{This is without loss of generality, because we can always allow the set of outcomes to be a set of probability distributions.} We can think of functions $p \colon X \to R$ as representing any economic \emph{situation} or \emph{context} we can put the player into. A player's goals will be defined on such functions. 

\begin{example} Consider our first example from the Introduction: three judges are voting simultaneously for one of two contestants $X = \{ A , B \}$. The winner is decided by the majority rule $\maj \colon X \times X \times X \to X$. In a setting where judges 1 and 3 have fixed their choices, say $x_1 = A$ and $x_3 = B$, this gives rise to a game context for the second judge, namely
\[ x_2 \quad \mapsto \quad \maj(A, x_2, B) \]
which is in fact the identity function since $\maj(A, x_2, B) = x_2$. If, on the other hand, judges 1 and 3 had fixed their choices as $x_1 = x_3 = A$, the game context for player 2 would be the constant function $x_2 \mapsto A$, since his vote does not influence the outcome. 
\end{example}

One can think of the game context $p : X \to R$ as an \emph{abstraction} of the actual game context that is determined by knowing the rules of the game, and how each opponent played. Notice that in the example above the game context which maps $A$ to $B$, and $B$ to $A$, never arises. It would arise, however, if one replaced the majority rule by the minority one. The lesson here is that several concrete contexts can give rise to the same abstract game context, and there might abstract game contexts that do not arise in a given particular game.

It might seem like we are losing too much information by adopting such an abstraction of a game context. What we wish to illustrate through several examples is that such a level of abstraction is sufficient for modeling players' individual motivations and goals. And precisely because it is abstract and it captures the strategic context of a player as if it was a single decision problem, it allows for a description of the players' intrinsic motivations, irrespective of how many players are around, or which particular game is being played. 

%---------------------------------------
\subsection {Quantifiers and Selection Functions}
\label{sec:quantifier_preference_util}
%---------------------------------------

Suppose now that $\mathcal A$ makes a decision $x \in X$ in a game context $p \colon X \to R$. The player will consider some outcomes to be \emph{good} (or \emph{acceptable}), and other outcomes to be bad. In general, we are going to allow the set of outcomes that the player considers good to be totally arbitrary. 

\begin{definition}[Quantifiers, \cite{escardo10a,escardo_sequential_2011}] Let $\power{R}$ denote the power-set of the set of outcomes $R$. Higher-order functions $$\varphi : (X \to R) \to \power{R}$$ from contexts $p : X \to R$ to sets of outcomes $\varphi(p) \subseteq R$ are called \emph{quantifiers}\footnote{The terminology comes from the observation that the usual existential $\exists$ and universal $\forall$ quantifiers of logic can be seen as operations of type $(X \to \BB) \to \BB$, where $\BB$ is the type of booleans. Mostowski \cite{Mostowski_1957} also called arbitrary functionals of type $(X \to \BB) \to \BB$ \emph{generalized quantifiers}. We are choosing to generalize this further by replacing the booleans $\BB$ with an arbitrary type $R$, and allowing for the operation to be multi-valued.}.
\end{definition}

We wish to model players $\mathcal A$ as quantifiers $\varphi_{\mathcal A} \colon (X \to R) \to \power{R}$. We think of $\varphi(p)$ as the set of outcomes the player $\mathcal A$ considers preferable in a given game context $p \colon X \to R$. It is crucial to recognize that this is a \emph{qualitative} description of a player, in the sense that an outcome is either preferable or it is not, with no numerical measure attached.

The classical example of a quantifier is \emph{utility maximisation}, with the outcome set $R = \RR^n$ consisting of $n$-tuples of real-valued payoffs. If we denote by $\pi_i \colon \RR^n \to \RR$ the $i$-projection, then the utility of the  $i^{th}$ player is $\pi_i(r)$.
Hence, given a game context $p : X \to \mathbb R^n$, the good outcomes for the $i^{th}$ player are precisely those for which the  $i^{th}$ coordinate, i.e. his utility, is maximal. This quantifier is given by
\[ \imax{i}(p)= \{ r \in \Img(p)\mid r_i \geq (\pi_i \circ p)(x') \text{ for all } x' \in X \} \]
where $\Img(p)$ denotes the image of the function $p \colon X \to R$, and $\pi_i \circ p$ denotes the composition of $p$ with the $i$-th projection. 

Analogously to utility maximization, players' choices motivated by preference relations can also be easily represented.  Suppose $R$ is the set of possible final outcomes, 
the context $p \colon X \to R$ maps actions to outcomes, and a player $i$ has a total order relation $\succeq_i$ on $R$, so that $x \succeq_i y$ means that player $i$ prefers the outcome $x$ to $y$. This quantifier is given by:
\[ \imax{\succeq_i}(p)= \{ r \in \Img(p)\mid r \succeq_i p(x') \text{ for all } x' \in X \} \]

Just as a quantifier tells us which outcomes a player considers good in each given context, one can also consider the higher-order function that determines which \emph{moves} a player considers good in any given context. 

\begin{definition}[Selection functions] A \emph{selection function} is any function of the form\footnote{In the computer science literature where selection functions have been considered previously \cite{escardo10a,escardo_sequential_2011} the focus was on single-valued ones. However, as multi-valued selection functions are extremely important in our examples we have adapted the definitions accordingly.}
\[ \varepsilon : (X \to R) \to \mathcal P (X). \]
\end{definition}

Similarly to quantifiers, the canonical example of a selection function is maximising one of the coordinates in $\mathbb R^n$, defined by
\[ \iargmax{i}(p)= \{ x \in X \mid (\pi_i \circ p)(x) \geq (\pi_i \circ p)(x') \text{ for all } x' \in X \}. \]
Even in one-dimensional $\mathbb R^1$ the $\argmax$ selection function is naturally multi-valued: a function may attain its maximum value at several different points.

%---------------------------------------
\subsection {Closed Selection Functions}
%---------------------------------------

The selection function one obtains from utility functions and preference relations in the previous subsection are examples of what we call  \emph{closed selection functions}.

\begin{definition}[Closedness] A selection function $\varepsilon : (X \to R) \to \mathcal P (X)$ is said to be \emph{closed} if whenever $x \in \varepsilon(p)$ and $p(x) = p(x')$ then $x' \in \varepsilon(p)$.
\end{definition}

Intuitively, a closed selection function is one which chooses optimal moves only based on the outcomes they generate. Two moves that lead to the same outcome are therefore indistinguishable, they are either both good or bad. It is easy to see that the selection function $\argmax(p)$ is closed. Agents modelled via closed selection functions do not put any preferences on moves that lead to identical outcomes. 

An example of non-closed selection function is the fixpoint operator
\[ \fix : (X \to X) \to \mathcal P (X). \]
Recall that a fixpoint of a function $f : X \to X$ is a point $x \in X$ satisfying $f(x) = x$. When the set of moves is equal to the set of outcomes $R = X$ there is a selection function whose good moves are precisely the fixpoints of the context. If the context has no fixpoint then the player will be equally satisfied with any outcome. Therefore such a selection function can be defined as
\[
\fix (p)= \begin{cases}
\{ x \in X \mid p(x) = x \} &\text{ if $p(x) = x$ for some $x \in X$ } \\
X &\text{ otherwise}
\end{cases}
\]
Clearly $\fix(\cdot)$ is non-closed, since we might have two points $x$ and $x'$ which both map to $x$ (i.e. $p(x) = p(x') = x$) so that $x$ is a fixed point but $x'$ might not be. We will see in Section \ref{sec:keynesian-beauty-contest} that such a fixpoint selection function perfectly models the agent in the Keynesian beauty contest whose sole goal is to vote for the winner.

%%%%%%%%%%%%%%%%%%%%%
\subsection{Relating Quantifiers and Selection Functions}
%%%%%%%%%%%%%%%%%%%%%

It is clear that quantifiers and selection functions are closely related. One important relation between them is that of \emph{attainment}. Intuitively this means that the outcome of a good move should be a good outcome. 

\begin{definition} Given a quantifier $\varphi : (X \to R) \to \mathcal P (R)$ and a selection function $\varepsilon : (X \to R) \to \mathcal P (X)$, we say that $\varepsilon$ \emph{attains} $\varphi$ iff for all contexts $p : X \to R$ it is the case that
\[ x \in \varepsilon (p) \implies p (x) \in \varphi (p). \]
\end{definition}

One can check that the attainability relation holds between the quantifier $\imax{i}$ and the selection function $\iargmax{i}$. Any point where the maximum valued is attained will evaluate to the maximum value of the function. More interestingly, the fixpoint quantifier is also a selection function, and it attains itself since
\[ x \in \fix (p)\implies p (x) \in \fix(p). \]

Let us briefly reflect on the game theoretic meaning of attainability. Suppose we have a quantifier $\varphi$ which describes the outcomes that a player considers to be good. The quantifier might be \emph{unrealistic} in the sense that it has no attainable good outcome. For example, a player may consider it a good outcome if he received a million dollars, but in his current context there may just not be a move available which will lead to this outcome. Given a context $p$, the set of attainable outcomes is precisely the image of $p$. 

Given any selection function $\varepsilon \colon (X \to R) \to \power{X}$, we can form the smallest quantifier which it attains as follows.

\begin{definition} \label{over-sel} Given a selection function $\varepsilon \colon (X \to R) \to \power{X}$, define the quantifier $\overline{\varepsilon} \colon (X \to R) \to \power{R}$ as
\[ \overline{\varepsilon} (p)= \{ p (x) \mid x \in \varepsilon (p)\}. \]
\end{definition}

Conversely, given any quantifier $\varphi \colon (X \to R) \to \power{R}$ we can define a corresponding selection function as follows. 

\begin{definition} \label{over-quant} Given a quantifier $\varphi \colon (X \to R) \to \power{R}$, define the selection function $\overline{\varphi} \colon (X \to R) \to \power{X}$ as
\[ \overline{\varphi}(p) = \{ x \mid p(x) \in \varphi(p) \}. \]
\end{definition}

We use the same overline notation, as it will be clear from the setting whether we are applying it to a quantifier or a selection function.
One can consider translating quantifiers into selection functions and back into quantifiers, or conversely.

\begin{proposition} \label{prop:closure} For all $p \colon X \to R$ we have $\overline{\overline{\varphi}}(p) = \varphi(p)$ and $\varepsilon(p) \subseteq \overline{\overline{\varepsilon}}(p)$.
\end{proposition}
\noindent {\bf Proof}. These are easy to derive. Let us briefly outline $\varepsilon(p) \subseteq \overline{\overline{\varepsilon}}(p)$. Suppose $x \in \varepsilon(p)$ is a good move in the game context $p \colon X \to R$. By Definition \ref{over-sel} we have that $p(x) \in \overline{\varepsilon}(p)$. Finally, by Definition \ref{over-quant} we have that $x \in \overline{\overline{\varepsilon}}(p)$. $\Box$ \\

The proposition above shows that on quantifiers the double-overline operation calculates the same quantifier we started with. However, on selection functions the mapping $\varepsilon \mapsto \overline{\overline{\varepsilon}}$ can be viewed as a \emph{closure} operator\footnote{Note that we might have a strict inclusion $\varepsilon(p) \subset \overline{\overline{\varepsilon}}(p)$ in case we have $x_1 \neq x_2$, with $x_1 \in \varepsilon(p)$ and $x_2 \not\in \varepsilon(p)$ but $p(x_1) = p(x_2)$. }. Intuitively, the new selection function $\overline{\overline{\varepsilon}}$ will have the same good \emph{outcomes} as the original one, but it might consider many more \emph{moves} to be good as well, as it does not distinguish moves which both lead to equally good outcomes.

\begin{proposition} A selection function $\varepsilon$ is \emph{closed} if and only if $\varepsilon = \overline{\overline{\varepsilon}}$.
\end{proposition}

\begin{proof} Assume first that $\varepsilon$ is closed, i.e. 
\begin{itemize}
	\item[(i)] $x \in \varepsilon(p)$ and $p(x) = p(x')$ then $x' \in \varepsilon(p)$.
\end{itemize}
By Proposition \ref{prop:closure} is it enough to show that if $x' \in \overline{\overline{\varepsilon}}(p)$ then $x' \in \varepsilon(p)$. Assuming $x' \in \overline{\overline{\varepsilon}}(p)$, and by Definition \ref{over-quant} we have
\begin{itemize}
	\item[(ii)] $p(x') \in \overline{\varepsilon}(p)$.
\end{itemize}
By Definition \ref{over-sel}, (ii) says that $p(x') = p(x)$ for some $x \in \varepsilon(p)$. By (i) it follows that $x \in \varepsilon(p)$. \\
Conversely, assume that $\varepsilon = \overline{\overline{\varepsilon}}$ and that $x \in \varepsilon(p)$ and $p(x) = p(x')$. We wish to show that $x' \in \varepsilon(p)$. Since $x \in \varepsilon(p)$ then $p(x) \in \overline{\varepsilon}(p)$. But since $p(x) = p(x')$ we have that $p(x') \in \overline{\varepsilon}(p)$. Hence, $x' \in \overline{\overline{\varepsilon}}(p)$. But since $\varepsilon = \overline{\overline{\varepsilon}}$ it follows that $x' \in \varepsilon(p)$.
\end{proof}

\begin{remark} The theory of quantifiers and selection functions has been developed in stages. Single-valued selection functions and quantifiers in the general form used here first appeared in \cite{escardo10a}, unifying earlier definitions in proof theory and type theory. That is also where the connection between selection functions and game theory was first established. Multi-valued quantifiers appeared in \cite{escardo_sequential_2011}, which allows us to capture more important examples in a more natural way. The connections between selection functions and game theory were explored in more depth in \cite{escardo12} and \cite{hedges13}, and the latter contains the definition of attainment given here. Finally \cite{hedges14a} contains the terminology \emph{context}.
\end{remark}

%---------------------------------------
\section{Higher-Order Games} 
%---------------------------------------
\label{sec:games}

Quantifiers and selection functions as introduced in the previous section can be used to model games. 
In this section we define higher-order games. In the next section we introduce suitable equilibrium concepts.

\begin{definition}[Higher-Order Games] \label{def-game-general} An $n$-players game ${\cal G}$, with a set $R$ of outcomes and sets $X_i$ of strategies for the  $i^{th}$ player, consists of an $(n+1)$-tuple ${\cal G} = (\varepsilon_1, \ldots, \varepsilon_n, q)$ where
\begin{itemize}
\item for each player $1 \leq i \leq n$, $\varepsilon_i : (X_i \to R) \to \mathcal P (X_i)$ is a selection function describing the $i$-th player's preferred moves in each game context.
\item $q : \prod_{i =1 }^n X_i \to R$ is the outcome function, i.e., a mapping from the strategy profile to the final outcome.
\end{itemize}
\end{definition}

Intuitively, we think of the outcome function $q$ as representing the `situation', or the rules of the game, while we think of the selection functions as describing the players. Thus we can imagine the same player in different situations, and different players in the same situation. This allows us to decompose a modelling problem into a global and a local part: modelling the situation and modelling the players.

\begin{remark}[Classical Game \cite{Osborne1994}] \label{normal-form-remark} The ordinary definition of a normal form game of $n$-players with standard payoff functions is a particular case of Definition \ref{def-game-general} when
\begin{itemize}
\item for each player $i$ the set of strategies is $X_i$,
\item the set of outcomes $R$ is $\mathbb R^n$, modelling the vector of payoffs obtained by each player,
\item the selection function of player $i$ is $\iargmax{i} \colon (X_i \to \RR^n) \to {\mathcal P}(X_i)$, i.e. $\argmax$ with respect to the  $i^{th}$ coordinate, representing the idea that each player is solely interested in maximising their own payoff,
\item the  $i^{th}$ component of the outcome function $q \colon \prod_{i=1}^n X_i \to \RR^n$ can be viewed as the payoff function $q_i \colon \prod_{j =1 }^n X_j \to \RR$ of the  $i^{th}$ player.
\end{itemize}
\end{remark}

%---------------------------------------
\subsection{Keynesian Beauty Contest}
\label{sec:keynesian-beauty-contest}
%---------------------------------------

Consider the following voting contest as already outlined in the introduction.

\begin{example} [The voting contest] \label{keynesian-ex1} 
There are three players, the judges $J=\{J_1, J_2, J_3\}$, who each vote for one of two contestants $A$ or $B$. The winner is determined by the simple majority rule.
\end{example}

We analyse two instances of this game with different motivations of players while keeping the overall structure of the game fixed.

\paragraph{Classical game.} 

First suppose that the judges rank the contestants according to a preference ordering. For example, say that judges 1 and 2 prefer $A$ and judge 3 prefers $B$. Table \ref{tab:classical} depicts a payoff matrix which encodes this situation, including the rules for choosing a player (majority) and the goals of each individual player.

\begin{table}[!htb]\label{tab:classical}
      \begin{minipage}{.5\linewidth}
                  \centering
                  \begin{tabular}{l|ll}
                    & A & B   \\
                    \hline
                  A & 1, 1, 0  & 1, 1, 0 \\
                  B & 1, 1, 0 & 0, 0, 1 
                  \end{tabular}
                  \footnotesize
                  \emph{$A$} 
      \end{minipage}%
      \begin{minipage}{.5\linewidth}
                 \centering
                 \begin{tabular}{l|ll}
                    & A & B  \\
                    \hline
                  A & 1, 1, 0 & 0, 0, 1 \\ 
                  B & 0, 0, 1 & 0, 0, 1 \\
                 \end{tabular}
                 \footnotesize
                 \emph{$B$}
      \end{minipage} 
      \caption{Voting Contest}
\end{table}

Let us now see how such a game would be modelled following Definition \ref{def-game-general}. The set of strategies in this case is the same as the set of possible outcomes, i.e. $X_i = R = \{ A, B \}$. The outcome function $q \colon X_1 \times X_2 \times X_3 \to R$ is nothing more than the majority function $\operatorname{maj}: X \times X \times X \rightarrow X$, e.g. $\maj(A, B, B) = B$. It remains for us to find suitable selection functions representing the goals of the three players. Consider two order relations on $X$, call it $B \preceq' A$ and $A \preceq'' B$. It is clear that in this case the judges wish to maximise the final outcome with respect to their preferred ordering. Hence we have that the three selection functions should be
\begin{align*}
\varepsilon_1(p) = \varepsilon_2(p) &= \iargmax{\preceq'} \\
\varepsilon_3 (p)&= \iargmax{\preceq''}.
\end{align*}
Therefore, the game above would be fully described by the tuple of higher-order functionals $${\cal G} = (\iargmax{\preceq'}, \iargmax{\preceq'}, \iargmax{\preceq''}, \maj).$$

\paragraph{The Keynesian variant.}
%%%%%%%%%%%%%%%%%%%%%%%%%

Let us now reconsider the beauty contest from the introduction. The first judge $J_1$ still ranks the candidates according to a preference ordering $B \preceq A$. The second and third judges, however, have no preference relations over the candidates per se, but want to vote for the winning candidate. They are Keynesian players! 

As shown in the introduction, it is perfectly possible to model such game via standard payoff matrices, and Table \ref{tab:voting_2} presents such an encoding. If there is a majority for a candidate and player $J_2$ or $J_3$ vote for the majority candidate they will get a certain utility, say 1. If they vote for another candidate, their utility is lower, say 0.  

\begin{table}[!htb]\label{tab:voting_2}
      \begin{minipage}{.5\linewidth}
                  \centering
                  \begin{tabular}{l|ll}
                    & A & B   \\
                    \hline
                  A & 1, 1, 1 & 1, 0, 1 \\
                  B & 1, 1, 1 & 0, 1, 0 
                  \end{tabular}
                  \footnotesize
                  \emph{$A$} 
      \end{minipage}%
      \begin{minipage}{.5\linewidth}
                 \centering
                 \begin{tabular}{l|ll}
                    & A & B  \\
                    \hline
                  A & 1, 1, 0 & 0, 1, 1 \\ 
                  B & 0, 0, 1 & 0, 1, 1 \\
                 \end{tabular}
                 \footnotesize
                 \emph{$B$}
      \end{minipage} 
      \caption{Voting Contest}
\end{table}

Note, however, that in the process of attaching utilities to the strategies, one has to compute the outcome of the votes, then check for the second and the third player whether their vote is in line with the outcome, and finally attach the utilities. 

Let us now contrast this with the higher-order modelling of games. First note that from the game ${\cal G}$ of the previous example, only the ``motivation" of players 2 and 3 have changed. Accordingly, we will only need to adjust their selection functions so as to capture their new goal which is to vote for the winner of the contest. Such goal is exactly captured by equipping $J_2$ and $J_3$ with the \emph{fixpoint selection function} $\fix : (X \to X) \to \mathcal P (X)$, defined in Section \ref{sec:quantifier_preference_util}. Note that it is neither necessary to change the structure of the game nor to manually compute anything. The new game with the two Keynesian judges is directly described by the tuple
$${\cal G}_{\textup K} = (\iargmax{\preceq}, \fix, \fix, \maj).$$
One can say that in the higher-order modelling of games we have equipped the individual players themselves with the problem solving ability that we used to compute the payoff matrices such that they represent the motivations of the Keynesian players.
These fixpoint agents with their computational power resemble a construction that is at the core of the 
Lucas critique.\footnote{Sargent \cite{Sargent1993} describes the need for a similarity of the economist and the economically reasoning agents 
in the economists' models as follows:
{\it  ``[t]he idea of rational expectations is ... said to embody the idea that economists 
and the agents they are modeling should be placed on the  equal footing: 
the agents in the model should be able to forecast and profit-maximize 
and utility-maximize as well as the economist - or should we say the econometrician - who constructed the model."}}
%---------------------------------------
\subsection{Meeting in New York} 
%---------------------------------------

Consider now the game where two strangers, call them $1$ and $2$, want to meet in New York. Suppose there are two places which they consider as meeting points, Grand Central Terminal ($G$) and the Empire State Building ($E$). Both players have to choose simultaneously, and they would only be happy if they pick the same meeting place. Table \ref{tab:Meeting_New_York} depicts a payoff matrix representing this game \cite{schelling1960strategy,Mas-Colell1995}. 
\begin{table}[!htb]
\begin{center}
\begin{tabular}{|c|c|c|c|l|}\hline
Strategy 	& $E$   	& $G$      \\ \hline \hline
$E$		& 1,1      	& 0,0    \\ \hline
$G$		& 0,0       	& 1,1    \\ \hline
\end{tabular}
\end{center}
\label{tab:Meeting_New_York}
\caption{Meeting in New York}
\end{table}

Notice that the original motivation of the players that is encoded in the numerical payoffs above is actually qualitative: Their only goal is to coordinate. Numeric values such as $1$ and $0$ encode whether they succeed or not. 

Let us consider a possible higher-order modelling of such game and see how it is able to directly express the motivation of the players. Let $C = \{E , G\}$ be the set of choices.  To begin with, the strategy spaces of the two players are  $X_1 = X_2 = C$. As the set of outcomes let us take the possible locations that the two players chose, i.e. $R = C \times C$, so an element of $R$ is a pair where the first coordinate tells what the first player chose, and the second tells what the second player chose. Now the outcome function $q : X_1 \times X_2 \to R$ is simply the identity function ${\rm id} \colon C \times C \to C \times C$.

Finally, to model the coordination goals of both player we can use a variant of the fixed point selection function $\fix_i : (C \to C \times C) \to {\mathcal P}(C)$, for $i \in \{1, 2\}$, defined as
\[\fix_i(p)= \begin{cases}
\{ x \in X \mid (\pi_i \circ p)(x) = x \} &\text{ if $(\pi_i \circ p)(x) = x$ for some $x \in X$ } \\
X &\text{ otherwise}.
\end{cases} \]
The game then has a higher-order description as the triple ${\cal G}_{\textup NY} = (\fix_2, \fix_1, {\rm id})$. 

%---------------------------------------
\subsection{Matching Pennies} 
%---------------------------------------
Two players, call them $1$ and $2$, play by hiding a penny in their hand. Each player secretly turns the penny such that heads ($H$) or tails ($T$) is facing up. Then, both players simultaneously reveal their penny. The two players' payoffs are such that player $1$ wins if both pennies match, player $2$ wins if they do not match. In this game there is a clear quantitative interpretation of payoffs: the player who loses gives his penny to the winner. Hence, the winner gains one penny while the other player loses one penny. Table \ref{tab:Matching_pennies} shows how such a game would be normally encoded into a payoff matrix. 
\begin{table}[!htb]
\begin{center}
\begin{tabular}{|c|c|c|c|l|}\hline
Strategy 	& $H$	& $T$      \\ \hline \hline
$H$		& 1,-1        & -1,1    		\\ \hline
$T$    	& -1,1        & 1,-1 \\ \hline
\end{tabular}
\end{center}
\caption{Matching Pennies}
\label{tab:Matching_pennies}
\end{table}

Payoffs thus characterise the monetary gains and losses of players. According to this interpretation the numerical payoffs come very natural. But there are also alternative interpretations of this game, such as penalty kicks in soccer, where the payoffs just capture qualitative information about who loses and who wins. In this case what we want to model as the players' goals is \emph{coordination} and \emph{differentiation}, respectively. 

Let us see how to model this in the higher-order game framework. Note that the first player wants to coordinate with the second player whereas the second player wants to differentiate from player 1. Similarly to Meeting in New York, player 1 can be modelled with a fixed point selection function. We let the set of outcomes be $R = X_1 \times X_2 = \{ H, T \}^2$, representing the actual choice made by the two players. Again the outcome function is the identity function:
\[ q(x_1,x_2) = {\rm id}(x_1, x_2) = (x_1,x_2).\]
As in the previous example, the first player, who wishes to coordinate the moves, is best modelled by the fixed point selection function
\[ \fix (p) = \{ x_1 \in X_1 \,\mid\, x_1 =\pi_2(p(x_1)) \}. \]
Just as the fixpoint selection function models coordination, so there is a `non-fixpoint' selection function which models differentiation or anti-coordination. The set of non-fixpoints in this game is 
\[ \nonfix(p)= \{ x_2 \in X_2 \mid x_2 \neq \pi_1(p(x_2)) \}. \]
So the whole game is modelled by the triple of higher order functions
\[ {\cal G}_{\textup MP} = (\fix, \nonfix, {\rm id}). \]

Coordination and anti-coordination are motifs that play an important role in general. Often, as for instance in the Game of Chicken, the situation we actually want to represent is qualitative. Fixpoints and non-fixpoints are a high-level description of such motivations.   
%---------------------------------------
\subsection{Battle of the Sexes} 
%---------------------------------------

A couple has agreed to meet but they do not agree whether they should be attending the ballet ($B$) or a football match ($F$). As far as individual preferences go, the husband prefers football over ballet, while the wife prefers ballet over football. But irrespective of their personal preferences, they would of course rather be together than by themselves in two different places. The set of strategies contain the two possible choices they can make $X_h = X_w =
\{B,F\}$. This is a variant of the classical Battle of the Sexes game. It is normally modelled as having an outcome $R = \RR \times \RR$ capturing the utility of the husband and wife after their individual choices are made. This game is represented via the payoff matrix of Table \ref{tab:BattleofSexes-classic}. 

\begin{table}[t!]
\begin{center}
\begin{tabular}{|c|c|c|c|l|}\hline
Strategy 	& $B$   	& $F$      	\\ \hline \hline
$B$ 	    	& 3,2       	& 1,1     	\\ \hline
$F$ 		& 0,0       	& 2,3        	\\ \hline
\end{tabular}
\end{center}
\caption{Battle of the Sexes}
\label{tab:BattleofSexes-classic}
\end{table}

In the classical representation, both players maximise their corresponding coordinate of the outcome tuple $(r_w, r_h) \in \RR \times \RR$, i.e. the wife wants to maximise $r_w$ whereas the husband would like to maximise $r_h$. 

Let us again consider an alternative modelling of the game using selection functions. The choices of moves are still $X_w = X_h = \{ B, F \}$, but now we take as the set of outcomes a description of what actually happens, namely, who goes to which event. We set $R = X_w \times X_h = \{ B, F \} \times \{ B, F \}$. Hence, an element of $R$ is a pair where the first coordinate denotes the choice of the wife and the second coordinate denotes the choice of the husband. Here again, we can take the outcome function $q : X_w \times X_h \to R$ to be the identity function.

We will build the selection functions for each player in a compositional way, by observing that each player has a lexicographic preference: their first priority is to be coordinated, and all else being equal, their second priority is to go to their favourite event (be it ballet or football). We describe an element of $R$ as coordinated if its first coordinate equals its second coordinate, so the coordinated outcomes in this game are $(B,B)$ and $(F,F)$. We can define a selection function $\varepsilon_c$ which chooses all moves that lead to a coordinated outcome as follows:
\[ \varepsilon_c (p) = \{ x \; | \; (\pi_1 \circ p)(x) = (\pi_2 \circ p)(x) \}). \]
Next we have a pair of selection functions $\varepsilon_b, \varepsilon_f$ representing the purely selfish aims of attending ballet and football respectively:
\[ \varepsilon_{B} (p) = \{ x \; | \; (\pi_1 \circ p)(x) = B \} \]
\[ \varepsilon_{F} (p) = \{ x \; | \; (\pi_2 \circ p)(x) = F \} \]
We are now ready to build our players' selection functions compositionally. Given a context, the joint selection function checks whether there are any moves which satisfy both personalities given by the coordinating and selfish selection functions. If so, the joint selection function returns those moves. If there are no moves satisfying both then the coordination takes priority, and the selfish aspect is ignored. Therefore the wife's selection function is
\[ \varepsilon_w (p) = \begin{cases}
\varepsilon_c (p) \cap \varepsilon_{B} (p) & \text{ if nonempty } \\
\varepsilon_c (p) & \text{ otherwise }
\end{cases} \]
and the husband's selection function is
\[ \varepsilon_h (p) = \begin{cases}
\varepsilon_c (p) \cap \varepsilon_{F} (p) & \text{ if nonempty } \\
\varepsilon_c (p) & \text{ otherwise }
\end{cases} \]
The whole game is then fully described by the triple ${\cal G} = (\varepsilon_w, \varepsilon_h, {\rm id})$.

Observe that in building these selection functions we have \emph{not} made use of the assumption that the game's outcome function is the identity function. The selection functions will still describe the intrinsic motivations of the players even if we change the rules of the game. For example, say that the couple have an agreement that if the husband ever goes to the football match and the wife is not there, then he has to make his way to the ballet instead and meet her there. That changes the outcome function of the game as
\[ q(x_w, x_h) = \begin{cases}
(x_w, B) & \text{ if $x_w = B$ and $x_h = F$} \\
(x_w, x_h) & \text{ otherwise.}
\end{cases} \]
This new game is modelled by the triple ${\cal G} = (\varepsilon_w, \varepsilon_h, q)$.

Although this is a trivial example, we believe that this method of modelling will distinguish itself in its ability to scale easily to very complex situations. A realistic player may have many competing aims, some of these might be best modelled by closed or non-closed selection functions (for example immediate profit, long-term profit, fairness concerns, environmental concerns). Using selection functions allows us to treat each aim individually, and then afterwards combine them (with rules for breaking ties, such as the lexicographic rule in this example) into a realistic description of the player.

%---------------------------------------
\section{Higher-Order Equilibria}\label{sec:equilibria}
%---------------------------------------

In the following we turn to equilibrium concepts. We will introduce two equilibrium definitions, one based on quantifiers and one based on selection functions, both in the spirit of classical Nash equilibria. 

%---------------------------------------
\subsection{Quantifier Equilibrium}
%---------------------------------------

Consider a game with $n$ players, and a strategy profile $\mathbf x \in \prod_{i=1}^n X_i$. Given an outcome function $q \colon \mathbf x \in \prod_{i=1}^n X_i \to R$, the game outcome resulting from this choice of strategy profile is $q (\mathbf x)$. We can describe the \emph{game context} in which player $i$ unilaterally changes his strategy as
\[ 
\mathcal U^q_i (\mathbf x)(x_i') = q (\mathbf x [i \mapsto x_i'])
\]
where $\mathbf x [i \mapsto x_i']$ is the tuple obtained from $\mathbf x$ by replacing the  $i^{th}$ entry of the tuple $\mathbf x$ with $x_i'$. Note that indeed $\mathcal U^q_i(\mathbf x)$ has type $X_i \to R$, the appropriate type of a game context for player $i$.

We call the $n$ functions $\mathcal U^q_i$ ($1 \leq i \leq n$) the \emph{unilateral maps} of the game. They were introduced in \cite{hedges13} in which it is shown that the proof of Nash's theorem amounts to showing that the unilateral maps have certain topological (continuity and closure) properties. The concept of a context was introduced later in \cite{hedges14a}, so now we can say that $\mathcal U^q_i (\mathbf x) : X_i \to R$ is the game context in which the  $i^{th}$ player can unilaterally change his strategy, so we call it a unilateral context.

Using this notation we can abstract the classical definition of Nash equilibrium to our framework. 

\begin{definition}[Quantifier equilibrium] \label{def-gen-nash} Given a game ${\cal G} = (\varepsilon_1, \ldots, \varepsilon_n, q)$, we say that a strategy profile $\mathbf x \in \prod_{i=1}^n X_i$ is in \emph{quantifier equilibrium} if \[ 
q (\mathbf x) \in \overline{\varepsilon_i} (\mathcal U^q_i (\mathbf x)) 
\]
for all players $1 \leq i \leq n$.
\end{definition}

As with the usual notion of Nash equilibrium, we are also saying that a strategy profile is in quantifier equilibrium if no player has a motivation to unilaterally change their strategy. This is expressed formally by saying that preferred outcomes, specified by the selection function when applied to the unilateral context, contain the outcome obtained by sticking with the current strategy. 

For illustration, we now want to give calculations of the quantifier equilibria for the voting contest game  $${\cal G} = (\iargmax{\preceq'}, \iargmax{\preceq'}, \iargmax{\preceq''}, \maj).$$ as described in Section \ref{sec:keynesian-beauty-contest} in the notation of quantifiers and unilateral contexts. We look at two possible strategy profiles: $BBB$ and $BBA$.

We claim that $BBB$ is in quantifier equilibrium. Note that $BBB$ has outcome $\maj(BBB)=B$. Let us verify this for player 1. The unilateral context of player 1 is
$$\mathcal U^{\maj}_1 (BBB)(x) = \maj(xBB) = B,$$
meaning that in the given context the outcome is $B$ no matter what player 1 chooses to play. The maximisation quantifier applied to such a unilateral context gives
$$\overline{\varepsilon_1}(\mathcal U^{\maj}_1 (BBB)) = \imax{\succeq_1}(p)(\mathcal U^{\maj}_1 (BBB))=\{B\},$$
meaning that, in the given context, player 1's preferred outcome is $B$. Hence, we can conclude by $\maj(BBB) = B \in \{B\} = \overline{\varepsilon_1}(\mathcal U^{\maj}_1 (BBB)(x))$ 
that $B$ is a quantifier equilibrium strategy for player 1. This condition holds for each player and allows us to conclude that $BBB$ is a quantifier equilibrium.

On the other hand, we show that $BBA$ is not in quantifier equilibrium. We have that
$$\maj(BBA) = B \notin \{A\} = \overline{\varepsilon_1}(\mathcal U^{\maj}_1 (BBA)).$$
since $\mathcal U^{\maj}_1 (BBA)(x) = \maj(xBA) = x$. In other words, the strategy profile $BBA$ gives rise to a game context $\mathcal U^{\maj}_1 (BBA)(x)$ where player $1$ has an incentive to change his strategy to $A$, so that the new outcome $\maj(ABA) = A$ is better than the previous outcome $B$.

This game has three quantifier equilibria: $\{AAA,AAB,BBB\}$. They are exactly the same as the Nash equilibria in the normal form representation (cf. Table \ref{tab:classical}). We will
discuss this coincidence in more detail in Section \ref{sec:relation_eq}.

%---------------------------------------
\subsection{Selection Equilibrium}\label{sec:context_Nash}
%---------------------------------------

The definition of quantifier equilibrium is based on quantifiers. However, we can also use selection functions directly to define an equilibrium condition.

\begin{definition}[Selection equilibrium] \label{def-gen-nash-context} Given a game ${\cal G} = (\varepsilon_1, \ldots, \varepsilon_n, q)$, we say that a strategy profile $\mathbf x \in \prod_{i=1}^n X_i$ is in \emph{selection equilibrium} if \[ 
x_i \in \varepsilon_i (\mathcal U^q_i (\mathbf x)) 
\]
for all players $1 \leq i \leq n$, where $x_i$ is the  $i^{th}$ component of the tuple $\mathbf x$.
\end{definition}

As in the previous subsection, let us illustrate the concept above using our classical voting contest from Section \ref{sec:keynesian-beauty-contest}. The set of selection equilibria is $\{AAA,AAB,BBB\}$, the same as the set of quantifier equilibria. We will give examples where the two notions differ in Section \ref{sec:relation_ex}.

Consider ${BBB}$ and the rationale for player $1$. As seen above, his unilateral context is
\[\mathcal U^{\maj}_1 (BBB)(x) = \maj(xBB) = B.\]
Hence, given this game context his selection function calculates
\[ \varepsilon_1 (\mathcal U^{\maj}_1 (BBB))  =\{ B\}\] 
As before, given that he is not pivotal, an improvement by switching votes is not possible. The same condition holds analogously for the other players. 

Let us now investigate the strategy profile ${BBA}$. The unilateral context is 
\[ \mathcal U^{\maj}_1 (BBA)(x) = \maj(xBA) = x.\]
Given this context, the selection function tells us that player $1$ would switch to $A$:
\[\varepsilon_1(\mathcal U^{\maj}_1 (BBA)) = \{A\}.\]
Hence, ${BBA}$ is not a selection equilibrium.

%------------------
\section{Relationship Between Equilibria}%-------------------

Our goal in this section is to show that selection equilibrium is a strict refinement of quantifier equilibrium. Moreover, we will show that for closed selection functions the two notions coincide. The obvious question then arises: which concept is more reasonable when games involve non-closed selection functions? We will provide several examples based on the voting contest to argue that in such cases selection equilibrium is the adequate concept. 

%------------------
\subsection{Selection Refines Quantifier Equilibrium}\label{sec:relation_eq}

%-------------------

\begin{theorem} \label{eq-inclusion} Every selection equilibrium is a quantifier equilibrium.
\end{theorem}
\begin{proof} Recall that by definition, for every context $p$ we have 
\[ x \in \varepsilon_i (p)\implies p(x) \in \overline{\varepsilon_i} (p) \]
since $\overline{\varepsilon_i} (p) = \{ p (x) \mid x \in \varepsilon_i (p) \}$. Assuming that $\mathbf x$ is a selection equilibrium we have
\[ x_i \in \varepsilon_i (\mathcal U^q_i (\mathbf x)) \]
Therefore
\[ \mathcal U^q_i (\mathbf x) (x_i) \in \overline{\varepsilon_i} (\mathcal U^q_i (\mathbf x)) \]
It remains to note that $\mathcal U^q_i (\mathbf x) ( x_i) = q (\mathbf x)$, because $\mathbf x [i \mapsto x_i] = \mathbf x$.
\end{proof}

However, for closed selection functions the two notions coincide:

\begin{theorem} \label{theorem-closed} If $\varepsilon_i = \overline{\overline{\varepsilon_i}}$, for $1 \leq i \leq n$, then the two equilibrium concepts coincide. 
\end{theorem}
\begin{proof} Given the previous theorem, it remains to show that under the assumption $\varepsilon_i = \overline{\overline{\varepsilon_i}}$ any strategy profile ${\mathbf x}$ in quantifier equilibrium is also in selection equilibrium. Fix $i$ and suppose ${\mathbf x}$ is such that
\[
q (\mathbf x) \in \overline{\varepsilon_i} (\mathcal U^q_i (\mathbf x)). 
\]
Since $\mathcal U^q_i (\mathbf x) ( x_i) = q (\mathbf x)$, we have
\[
\mathcal U^q_i (\mathbf x)(x_i) \in \overline{\varepsilon_i} (\mathcal U^q_i (\mathbf x)). 
\]
By the definition of $\overline{\overline{\varepsilon_i}}$ it follows that
\[
x_i \in \overline{\overline{\varepsilon_i}}(\mathcal U^q_i (\mathbf x)). 
\]
Therefore, since $\varepsilon_i = \overline{\overline{\varepsilon_i}}$, we obtain $x_i \in \varepsilon_i(\mathcal U^q_i (\mathbf x))$.
\end{proof}

The theorem above explains why in all the examples from the last section the strategy profiles that were quantifier equilibrium were the same as those in selection equilibrium. All the examples can be modelled with closed selection functions. Moreover, since $\argmax$ can be easily shown to be closed, in the classical modelling of games via maximising players, our two notions of equilibrium also coincide. The following theorem shows that they both indeed also coincide with the standard notion of Nash equilibrium.

\begin{theorem} \label{cor-classical} In a classical game (see Remark \ref{normal-form-remark}) the standard definition of Nash equilibrium and the equilibrium notions of Definitions \ref{def-gen-nash} and \ref{def-gen-nash-context} are equivalent. 
\end{theorem}

\begin{proof}
Suppose the set of outcomes $R$ is $\RR^n$ and that the selection functions $\varepsilon_i$ are $i$-$\argmax$, i.e. maximising with respect to  $i^{th}$ coordinate. Unfolding Definition \ref{def-gen-nash-context} and that of a unilateral context $\mathcal U^q_i (\mathbf x)$, we see that a tuple $\mathbf x$ is an equilibrium strategy profile if for all $1 \leq i \leq n$
\[ x_i \in i\textup{-}\argmax_{x \in X_i} q (\mathbf x [i \mapsto x]). \]
But $x_i$ is a point on which the function $p(x) = q (\mathbf x [i \mapsto x])$ attains its maximum precisely when $p(x_i) \in \max_{x \in X_i} p(x)$. Hence 
\[ q(\mathbf x) = q (\mathbf x [i \mapsto x_i]) = p(x_i) = \max_{x \in X_i} p(x) = \max_{x \in X_i} q (\mathbf x [i \mapsto x]) \]
which is the standard definition of a Nash equilibrium: for each player $i$, the outcome obtained by not changing the strategy, i.e. $q(\mathbf x)$, is the best possible amongst the outcomes when any other available strategy is considered, i.e. $\max_{x \in X_i} q (\mathbf x [i \mapsto x])$. 
\end{proof}

Theorem \ref{cor-classical} above shows that in the case of classical games the usual concept of a Nash equilibrium coincides with both the quantifier equilibrium and the selection equilibrium. On the other hand, for general games, Theorem \ref{eq-inclusion} proves that every selection equilibrium is a quantifier equilibrium
\[
\text{selection equilibria} \subsetneq\text{quantifier equilibria}
\]
In the following subsection we give several examples showing that the inclusion above is strict, i.e. that there are games where selection equilibrium is a strict refinement of quantifier equilibrium. By Theorem \ref{theorem-closed} these examples necessarily make use of players modelled by non-closed selection functions.

%------------------
\subsection{Examples of Equilibria}\label{sec:relation_ex}

\label{sec:fixpoint}
%------------------

Now we will explore the distinction between quantifier equilibria and selection
equilibria in more detail and we also relate it to the standard approach based on Nash equilibria. 

\paragraph{Keynesian beauty contest.}

We have discussed the representation of this game both in normal form as well as in higher-order functions. Here, we will turn to analyzing the equilibria of its higher-order representation
$${\cal G}_{\textup K} = (\iargmax{\preceq}, \fix, \fix, \maj).$$
from Section \ref{sec:keynesian-beauty-contest}. We begin with quantifier equilibria (see Table \ref{tab:fixfix}). These include the strategy profiles where players $J_2$ and $J_3$ are both coordinated but also profiles where either $J_2$ or $J_3$ is in the minority.

We illustrate the rationale for the strategy profile $AAB$ of the Keynesian player 3.
The outcome of $AAB$ is $\maj(AAB)=A$. The unilateral context of player 3 is 
\[ \mathcal U^{\maj}_3 (AAB)(x)=\operatorname{maj}(AAx)=A \]
meaning that the outcome is (still) $A$ if player 3 unilaterally changes from $B$ to $A$.
The fixed point quantifier applied to this context gives 
\[ \overline{\varepsilon_3}(\mathcal U^{\maj}_3 (AAB)) = \fix(\mathcal U^{\maj}_3 (AAB)) = \{A\} \]
meaning that $A$ is the outcome resulting from an optimal choice.
Hence, we can conclude by 
\[ \maj(AAB) = A \in \{A\} = \overline{\varepsilon_3}(\mathcal U^{\maj}_3 (AAB)) \]
that player 3 is happy with his choice of move $B$ according to the quantifier equilibrium notion. This already demonstrates the problem with the quantifier equilibrium notion, since player 1 has voted for $B$ but $A$ is the winner, so he should not be happy at all!

Now, let us turn to the selection equilibria. Table \ref{tab:fixfix} also contains the selection equilibria and it shows that they are a strict subset of the quantifier equilibria. 

\begin{table}[th!]
\begin{center}
\begin{tabular}{|c|c||c|c||c|c|}\hline
Strategy 	& Outcome 	& Quantifier 	Eq. 	    & Defects	& Selection Eq.	        & Defects 		\\ \hline \hline
$AAA$ 	    & $A$      	&  \checkmark   &		    & \checkmark		& 				\\ \hline
$AAB$      	& $A$      	&  \checkmark   &		    & -					& $J_3$			\\ \hline
$ABA$ 		& $A$ 		&  \checkmark	&		    & -					& $J_2$			\\ \hline
$ABB$ 		& $B$ 		&  \checkmark   &		    & \checkmark		& 				\\ \hline
$BAA$ 		& $A$ 		&  \checkmark   &		    & \checkmark		& 	 			\\ \hline
$BAB$ 		& $B$ 		&  -  		    & $J_1$	    & -					& $J_1$, $J_2$ 	\\ \hline
$BBA$ 		& $B$ 		&  - 		    & $J_1$	    & -					& $J_1$, $J_3$ 	\\ \hline
$BBB$ 		& $B$ 		&  \checkmark 	&		    & \checkmark		& 				\\ \hline
\end{tabular}
\end{center}
\caption{Players: max, fix, fix}
\label{tab:fixfix}
\end{table}

Consider again the strategy profile $AAB$, focusing on the third player. In the case of the selection equilibrium we have
$$B \notin \{A\} = \fix(\mathcal U^{\maj}_3 (AAB)) = \varepsilon_3 (\mathcal U^{\maj}_3 (AAB))$$
meaning that player 3 is not happy with his current choice of strategy $B$ with respect to the strategy profile $AAB$.

\begin{remark} Given Theorem \ref{theorem-closed} it follows immediately that $\fix \colon (X \to X) \to 2^X$ is not a closed selection function. Indeed, it is easy to calculate that
\[ \overline{\overline{\fix}}(p) = \{ x \; | \; p(x) = p(y), \mbox{ for some $y$ such that $y = p(y)$} \}, \]
i.e. $\overline{\overline{\fix}}(p)$ is the inverse image of $\fix(p)$, so it contains not only all fixed points of $p$ but also points that map through $p$ to a fixed point.
\end{remark}

The selection equilibria are precisely those in which $J_2$ and $J_3$ are coordinated, and $J_1$ is not pivotal in any of these. For illustration, consider the strategy $AAA$, which is a selection equilibrium of this game. Suppose the moves of $J_1$ and $J_2$ are fixed, but $J_3$ may unilaterally change strategy. The unilateral context is
\[ \mathcal U^{\maj}_3 (AAA) (x) = \operatorname{maj} (AAx) = A \]
Thus the unilateral context is a constant function, and its set of fixpoints is
\[ \fix (\mathcal U^{\maj}_3 (AAA)) = \{ A \} \]
This tells us that $J_3$ has no incentive to unilaterally change to the strategy $B$, because he will no longer be voting for the winner.

On the other hand, for the strategy $ABB$ the two Keynesian players are indifferent, because if either of them 
unilaterally changes to $A$ then $A$ will become the majority and they will still be voting for the winner. This is still a selection equilibrium (as we would expect) because the unilateral context is the identity function, and in particular $B$ is a fixpoint.

As a last point, let us compare the selection and quantifier equilibria of Table \ref{tab:fixfix} with the Nash equilibria in the normal form game. Table \ref{tab:voting_3} reproduces the payoff table; Nash equilibria payoffs are marked in bold. Note that the latter are the same as selection equilibria. 

The main message is that in the standard approach of modelling games via payoff matrices, it may become necessary to calculate payoffs by hand so as to model the given motivations of the player. In contrast, with selection functions, the goals of the player can be directly expressed
without any additional computation. And remarkably, the selection equilibria reproduce 
exactly the economic intuition carried by the payoffs that we encoded by hand.

%---------------------------------------
\paragraph{Coordination and Anti-Coordination.}
%---------------------------------------

\begin{table}[t]
      \begin{minipage}{.5\linewidth}
                  \centering
                  \begin{tabular}{l|ll}
                    & A & B   \\
                    \hline
                    A & \bf{1, 1, 1} & 1, 0, 1 \\
                    B & \bf{1, 1, 1} & 0, 1, 0 
                  \end{tabular}
                  \footnotesize
                  \emph{$A$} 
      \end{minipage}%
      \begin{minipage}{.5\linewidth}
                 \centering
                 \begin{tabular}{l|ll}
                    & A & B  \\
                    \hline
                    A & 1, 1, 0 & \bf{0, 1, 1} \\ 
                    B & 0, 0, 1 &\bf{0, 1, 1} \\
                 \end{tabular}
                 \footnotesize
                 \emph{$B$}
      \end{minipage} 
      \caption{Voting Contest}
	\label{tab:voting_3}
\end{table}

In the Meeting in New York as well as in the Matching Pennies examples, we have already seen that the fixpoint selection functions nicely capture coordination motifs and the anti-fixpoint function anti-coordination, respectively. Let us consider another variant of the voting game to reassure that selection equilibrium is the adequate concept to capture our intuition when selection functions are not closed.

We consider a game where all players want to vote for the winner of the contest. Table \ref{tab:voting_cor} represents the payoffs of this game; Nash equilibria are in bold. Clearly the only two equilibria are when all judges vote unanimously for a given contestant. Judges $J_1, J_2$ and $J_3$ want to vote for the winner, so the selection functions are given by the fixpoint operator $(X \to X) \to \power{X}$. 

\begin{table}[!htb]
      \begin{minipage}{.5\linewidth}
                  \centering
                  \begin{tabular}{l|ll}
                    & A & B   \\
                    \hline
                    A & \bf{1, 1, 1} & 1, 0, 1 \\
                    B & 0, 1, 1 & 1, 1, 0 
                  \end{tabular}
                  \footnotesize
                  \emph{$A$} 
      \end{minipage}%
      \begin{minipage}{.5\linewidth}
                 \centering
                 \begin{tabular}{l|ll}
                    & A & B  \\
                    \hline
                    A & 1, 1, 0 & 0, 1, 1 \\ 
                    B & 1, 0, 1 &\bf{1, 1, 1} \\
                 \end{tabular}
                 \footnotesize
                 \emph{$B$}
      \end{minipage} 
      \caption{Voting Contest}
      \label{tab:voting_cor}
\end{table}

As can be seen in Table \ref{tab:coordination}, the selection equilibria are exactly the coordinated strategies. This game is a good example of why quantifier equilibria are not suitable for modelling games with non-closed selection functions: it can be seen in the table that every strategy is a quantifier equilibrium of this game, but the selection equilibrium captures the economic intuition perfectly that the equilibria should be the strategy profiles that are maximally coordinated, namely $AAA$ and $BBB$.

\begin{table}[t!]
\begin{center}
\begin{tabular}{|c|c||c|c||c|c|}\hline
Strategy 	& Outcome 	& Quantifier           Eq.	 	& Defects	& Selection    Eq.  & Defects 	\\ \hline \hline
$AAA$ 	    & $A$  	    &  \checkmark    		&		    & \checkmark	& 	    	\\ \hline
$AAB$      	& $A$       &  \checkmark     		&		    & -				& $J_3$	    \\ \hline
$ABA$ 		& $A$ 		&  \checkmark      	    &		    & -				& $J_2$	    \\ \hline
$ABB$ 		& $B$ 		&  \checkmark       	&		    & -				& $J_1$	    \\ \hline
$BAA$ 		& $A$ 		&  \checkmark         	&		    & -				& $J_1$ 	\\ \hline
$BAB$ 		& $B$ 		&  \checkmark         	&		    & -				& $J_2$	    \\ \hline
$BBA$ 		& $B$ 		&  \checkmark         	&		    & -				& $J_3$ 	\\ \hline
$BBB$ 		& $B$ 		&  \checkmark         	&		    & \checkmark	& 	    	\\ \hline
\end{tabular}
\end{center}
\caption{Players: fix, fix, fix}
\label{tab:coordination}
\end{table}

\begin{table}[!thb]
      \begin{minipage}{.5\linewidth}
                  \centering
                  \begin{tabular}{l|ll}
                    & A & B   \\
                    \hline
                    A & 0, 0, 0 & \bf{0, 1, 0} \\
                    B & \bf{1, 0, 0} & \bf{0, 0, 1} 
                  \end{tabular}
                  \footnotesize
                  \emph{$A$} 
      \end{minipage}%
      \begin{minipage}{.5\linewidth}
                 \centering
                 \begin{tabular}{l|ll}
                    & A & B  \\
                    \hline
                    A & \bf{0, 0, 1} &\bf{1, 0, 0} \\ 
                    B & \bf{1, 0, 1} & 0, 0, 0 \\
                 \end{tabular}
                 \footnotesize
                 \emph{$B$}
      \end{minipage} 
      \caption{Voting Contest}
      \label{tab:voting_anti}
\end{table}

In a beauty contest, a player whose selection function is non-fixpoint is a `punk' who aims to be in a minority. To conclude, let us consider the game in which all three judges are punks (Table \ref{tab:voting_anti} represents the payoffs of this game; Table \ref{tab:anticoordination} represents the equilibria in the higher-order representation). Of course only one player can actually be in a minority, so the selection equilibria are precisely the `maximally anti-coordinated' strategy profiles, namely those in which one judge differs from the other two. This is another example of a game in which every
strategy is a quantifier equilibrium, but the selection equilibrium corresponds perfectly to our intuition.

\begin{table}[h!]
\begin{center}
\begin{tabular}{|c|c||c|c||c|c|}\hline
Strategy 	& Winner 	& Nash          Eq. 	 	& Defects	& Selection   Eq.    	& Defects       	\\ \hline \hline
$AAA$ 	    & $A$  	    &  \checkmark    		&		    & -					& $J_1$, $J_2$, $J_3$ \\ \hline
$AAB$      	& $A$      	&  \checkmark    		&		    & \checkmark		& 					\\ \hline
$ABA$ 		& $A$ 		&  \checkmark	       	&		    & \checkmark		& 					\\ \hline
$ABB$ 		& $B$ 		&  \checkmark      		&		    & \checkmark		& 					\\ \hline
$BAA$ 		& $A$ 		&  \checkmark        	&		    & \checkmark		& 					\\ \hline
$BAB$ 		& $B$ 		&  \checkmark           &		    & \checkmark		& 					\\ \hline
$BBA$ 		& $B$ 		&  \checkmark        	&		    & \checkmark		&  					\\ \hline
$BBB$ 		& $B$ 		&  \checkmark        	&		    & -					& $J_1$, $J_2$, $J_3$ \\ \hline
\end{tabular}
\end{center}
\caption{Players: non-fix, non-fix, non-fix}
\label{tab:anticoordination}
\end{table}

%---------------------------------------
\section{Conclusion}
%---------------------------------------
In this paper we have introduced an alternative representation of games. Quantifiers and selection functions provide an abstract and general way to describe players' goals. We have shown that, for instance in games of coordination where payoffs  are just implemented as a numerical representation of qualitative goals, higher-order functions provide a more direct way of describing these goals. What is more higher order functions formally summarise goals where utility functions do not exist - as
exemplified by the Keynesian Beauty contest. 

In an  accompanying paper on higher-order decision theory \cite{Hedges_et_al_2015_decisions} we show that quantifiers and selection functions provide a powerful tool to unify several approaches ranging from rational choice theory to behavioral alternatives. The game context, introduced in this paper,  basically summarises the game theoretic situation of each player as a decision problem. Hence, it is possible to consider behavioral alternatives such as heuristics not only in decision problems but they can be directly implemented in strategic situations as well.

In this paper we have restricted ourselves to show that a representation with higher-order functions does exist. A big advantage of our approach we have not really touched here is \emph{compositionality}. With payoffs it may be necessary to compute outcomes by hand. This is cumbersome, error prone, and slow. In particular, when one considers a change in the model such as adding another player. In contrast, selection functions can be used algebraically. There are clearly
defined ways to extend games without the need to change elements of the payoffs by hand. As the foundation of computer science, computability theory and functional programming languages, this compositionality is naturally built in higher-order functions. We believe this will be particularly helpful when analysing arbitrarily complicated and irregular games. 

The fact that there is very close connection to computation is another advantage of this framework. Our game representations are directly representable as code. In fact, we have taken advantage of this possibility and implemented a prototype in order to automatically compute the equilibria of the games in this paper.
%---------------------------------------
\bibliographystyle{plain}
\bibliography{../references}

%\newpage
%\tableofcontents

\end{document}